\documentclass[11pt]{article}
\pdfoutput=1
\usepackage{graphicx}
\usepackage{verbatim}
\usepackage{fullpage}
\usepackage{hyperref}
\usepackage{amssymb}
\usepackage{amsmath}
\usepackage{amsthm}

\def\final{0}  
\def\iflong{\iffalse}
\ifnum\final=0  
\newcommand{\knote}[1]{[{\tiny karthik: \bf #1}]\marginpar{*}}
\newcommand{\sidecomment}[1]{}
\else 
\newcommand{\knote}[1]{}
\newcommand{\sidecomment}[1]{}
\fi  

\newtheorem{theorem}{Theorem}

\newtheorem{lemma}[theorem]{Lemma}
\newtheorem{claim}[theorem]{Claim}

\def\reals{\mathbb{R}}

\def\E{\mathbb{E}}
\def\eps{\epsilon}
\def\grade{\gamma}

\newcommand{\abs}[1]{\left|#1\right|}
\newcommand{\prob}[1]{{\sf Pr}\left(#1\right)}
\newcommand{\expec}[1]{\operatorname{\mathbb{E}}\left(#1\right)}

\title{Finding a most biased coin with fewest flips}
\author{Karthekeyan Chandrasekaran \footnote{karthe@seas.harvard.edu, Harvard University. This work was done while the author was a visiting researcher at ICSI, Berkeley.}
\and 
Richard Karp \footnote{karp@icsi.berkeley.edu, University of California, Berkeley.}
}
\date{}

\begin{document}

\maketitle

\begin{abstract}
We study the problem of learning a most biased coin among a set of coins by tossing the coins adaptively. The goal is to minimize the number of tosses until we identify a coin $i^*$ whose posterior probability of being most biased is at least $1-\delta$ for a given $\delta$. Under a particular probabilistic model, we give an {\em optimal} algorithm, i.e., an algorithm that minimizes the expected number of future tosses. The problem is closely related to finding the best arm in the multi-armed bandit problem using adaptive strategies. Our algorithm employs an optimal adaptive strategy -- a strategy that performs the best possible action at each step after observing the outcomes of all previous coin tosses. Consequently, our algorithm is also optimal for any starting history of outcomes. To our knowledge, this is the first algorithm that employs an optimal adaptive strategy under a Bayesian setting for this problem. Our proof of optimality employs tools from the field of Markov games.

\end{abstract}


\section{Introduction}

The multi-armed bandit problem is a classical decision-theoretic problem with applications in bioinformatics, medical trials, stochastic algorithms, etc. \cite{MAB-book-GGW}. The input to the problem is a set of arms, each associated with an unknown stochastic reward. At each step, an agent chooses an arm and receives a reward. The objective is to find a strategy for choosing the arms in order to achieve the best expected reward asymptotically. This problem has spawned a rich literature on the trade off between exploration and exploitation while choosing the arms \cite{BF85, LR85,ACF02,ABFS03}. 

The motivation to identify the best bandit arm arises from problems where one would like to minimize regret within a fixed budget. In the models considered in \cite{BMS09, ABM10,GGLB11}, 
the goal is to choose an arm after a finite number of steps to minimize regret. Here regret is defined to be the difference between the expected reward of the chosen arm and the expected reward of the optimal arm. 
The work of \cite{BMS09} suggested that the exploration-exploitation trade offs for this setting are much different from the setting where the number of steps is asymptotic. Following this, Audibert et al. \cite{ABM10} proposed exploration strategies to perform essentially as well as the best strategy that knows all distributions up to permutations of the arms. Gabillon et al. \cite{GGLB11} addressed the problem of identifying the best arm for each bandit among a collection of bandits within a fixed budget. They proposed strategies that focus on arms whose expected rewards are closer to that of the optimal arm and show an upper bound on the probability of error for these strategies that decreases exponentially with the number of steps allowed. 

In contrast, one could also attempt to optimize the budget subject to the quality of the arm to be identified. This is identical to racing and action elimination algorithms \cite{hoeffding-races-MM94,action-elimination-EMM06,EMM02} which address the sample complexity of the pure exploration problem -- given any $\delta>0$, identify the arm with maximum expected reward with error probability at most $\delta$ while minimizing the total number of steps needed. This PAC-style learning formulation was introduced by Even-Dar et al. \cite{EMM02}. Given a collection of $n$ arms, Even-Dar et al. \cite{EMM02} showed that a total of $O((n/\eps^2)\log(1/\delta))$ steps is sufficient to identify an arm whose expected reward is at most $\eps$ away from the optimal arm with correctness at least $1-\delta$. Mannor and Tsitsiklis \cite{MT04} showed lower bounds matching up to constant factors under various settings of the rewards. We attempt to bridge the constant factor gap by addressing the problem from a decision-theoretic perspective. Given the history of outcomes, does there exist a strategy to choose an arm so that the expected number of steps needed to learn the best arm is minimized? Our notion of learning the best arm is to identify an arm whose posterior probability of being the most-rewarding is at least $1-\delta$.

Although the PAC-style learning problem appears to have garnered the interest of the learning theory community only over the past decade \cite{EMM02,MT04,C05,BMS09,ABM10,GGLB11}, it has been actively studied in the field of operations research for several decades as the ``ranking and selection problem''. It was introduced for normally distributed rewards by Bechhofer \cite{bechhofer54}. 
Adaptive strategies for this problem, known as ``sequential selection'', can be traced back to Paulson \cite{paulson64}. Variants of the problem find applications in minimizing the number of experimental simulations to achieve a given confidence level \cite{paulson64,book-ranking-selection-BSG95,kim-nelson,ranking-selection-simulation-BNK03,selection-for-simulation-PN01}. A simple and interesting case of the problem is when the most rewarding arm and the second-most rewarding arm differ in their mean rewards by at least $\eps>0$. This special case is known as the ``indifference-zone'' assumption \cite{bechhofer54}. Strategies and their measure of optimality are known for various relaxations of independence, normality, equal and known variances and indifference-zone assumptions \cite{kim-nelson}. In the Bayesian setting, the mean rewards of the normal distributions are chosen from some underlying distribution \cite{bayesian-R-S-GM96,bayesian-R-S-CI01,bayesian-R-S-FPD08,bayesian-R-S-CG09}. In this work, we address a particular Bayesian setting for Bernoulli rewards satisfying the indifference-zone assumption.

If the rewards from the bandit arms are Bernoulli, then learning the arm with the maximum expected reward is equivalent to learning the most biased coin by tossing them adaptively. So, we will focus on this problem for the rest of the paper. 
Under the indifference zone assumption, Chernoff bound leads to a trivial upper bound on the number of tosses in the non-adaptive setting -- toss each coin $(4/\eps^2)\log{(n/\delta)}$ times and output the coin with the maximum number of heads outcomes. Let $\hat{p_i}$ denote the empirical probability of heads for the $i$th coin. By Chernoff bound, $|\hat{p_i}-p_i|\le \eps/2$ with probability at least $1-\delta/n$. Therefore, by the union bound, it follows that this trivial toss-each-coin-$k$-times strategy outputs the most biased coin with probability at least $1-\delta$.

In this work, we give a simple yet optimal strategy for choosing coins to toss in a particular Bayesian setting. Our strategy is optimal in the sense that given a current history of outcomes of all coins and a threshold, it minimizes the expected number of tosses to find a coin whose posterior probability of being a most-biased coin is at or above the threshold. 
Our main contribution is a proof of optimality by employing tools from the field of Markov games. 
We also bound the expected number of coin tosses performed by our strategy. 
To the best of our knowledge, this is the first provably optimal strategy under a Bayesian setting of the problem with indifference zone assumption. \\ 


\noindent {\bf Setting.} A coin is said to be {\em heavy} if the probability of heads for the coin is $p+\eps$ and {\em not-heavy} if the heads probability is $p-\eps$ for some given $\eps \in (0,1/2)$ and $p\in[\eps,1-\eps]$. We are given an infinite collection of coins where each coin in the collection is heavy with probability $\alpha$ and not-heavy with probability $1-\alpha$. Given $\delta>0$, the algorithm is allowed to toss coins adaptively and has to necessarily perform a coin toss until it identifies a coin whose posterior probability of being heavy is at least $1-\delta$ (i.e., a coin $i$ for which $\prob{\text{Coin i is heavy }|\text{ Outcomes all coin tosses}}\ge 1-\delta$). The goal is to minimize the expected number of tosses required. 


An adaptive strategy is allowed to choose which coin to toss after observing the history of outcomes of all previous coin tosses. Given the history of outcomes of coin tosses, the cost of an adaptive strategy is equal to the expected number of future coin tosses needed by following this strategy so that it identifies a coin whose posterior probability of being heavy is at least $1-\delta$. An adaptive strategy is said to be optimal if it has the minimum cost. 

\subsection{Results}
Our main result is an optimal adaptive algorithm for the above setting.
\begin{theorem}\label{thm:algorithm-optimality}
Given $\delta>0$, there exists an algorithm $A$ that employs an optimal adaptive strategy in tossing coins
to identify a coin whose posterior probability of being heavy is at least $1-\delta$. 
At any step, the time taken by $A$ to identify the coin to toss is $O(1)$.
\end{theorem}

We also quantify the number of tosses performed by our optimal adaptive algorithm. We assume an infinite supply of coins under the same probabilistic setting. 
Let $q:=1-p$, $\Delta_H:=\log{((p+\eps)/(p-\eps))}$, $\Delta_T:=\log{((q+\eps)/(q-\eps))}$, $B(\delta):=\log{((1-\alpha)(1-\delta)/\alpha\delta)}$. Let $\delta_0$ be determined as follows: Consider the unique real value $\rho\in (0,1)$ such that $\rho^{\Delta_H}(p+\eps)+\rho^{-\Delta_T}(q-\eps)=1$ (the existence and uniqueness of $\rho$ is elaborated in Section \ref{sec:quantification}). Fix $\delta_0$ to be the largest real value such that $(1-\rho^{B(\delta)+\Delta_H})/(1-\rho^{B(\delta)+\Delta_T})<2$ and $B(\delta)\ge \Delta_H$.


\begin{theorem}\label{thm:no-of-tosses}
For every $\delta\in (0,\delta_0]$, 
the expected number of tosses performed by $A$ to identify a coin whose posterior probability of being heavy is at least $1-\delta$ in the above setting, is at most
\[
\frac{16}{\eps^2}\left(\frac{1-\alpha}{\alpha}+\log{\left(\frac{(1-\alpha)(1-\delta)}{\alpha \delta}\right)}\right).
\]
\end{theorem}

The implications of our upper bound when the number of coins is bounded but much larger than $1/\alpha$ needs to be contrasted with the lower bounds by \cite{MT04}. In this case, setting $n=c/\alpha$ in the above expression 
suggests that our algorithm beats the lower bound shown in Theorem 9 of \cite{MT04}.
We observe that Theorem 9 of \cite{MT04} shows a lower bound in the most general Bayesian setting -- there exists a prior distribution of the probabilities of the $n$ coins so that any algorithm requires at least $O((n/\eps^2)\log{(1/\delta)})$ tosses in expectation. However, our algorithm works in a particular Bayesian setting by exploiting prior knowledge about this setting.


\subsection{Algorithm}

At any stage of the algorithm, let the history of outcomes of a coin $i$ be given by $D_i:=(h_i,t_i)$ where $h_i$ and $t_i$ refer to the number of outcomes that were heads and tails respectively. Given the history $D_i$, we define the likelihood ratio of the coin to be
\[
L_i:=\frac{\prob{\text{Coin $i$ is heavy}|D_i}}{\prob{\text{Coin $i$ is not-heavy}|D_i}}
=\left(\frac{p+\eps}{p-\eps}\right)^{h_i} \left(\frac{q-\eps}{q+\eps}\right)^{t_i}.
\]
\begin{center}
\fbox{\parbox{5.8in}{
\begin{minipage}{5.8in}
\begin{tt}
\noindent {\bf Algorithm Likelihood-Toss}

\begin{enumerate}
\item Initialize $L_i=1$ for the $i$'th coin.
\item While 
($L_i<(1-\alpha)(1-\delta)/\alpha\delta\ \forall\ i\in [n]$)
\begin{enumerate}
\item Toss coin $i^*$ such that $i^*=\arg\max\{L_i:i\in [n]\}$. (Break ties arbitrarily). Let
\begin{equation*}
b_{i^*}=
\begin{cases}
1 & \text{if outcome is heads},\\
0 & \text{if outcome is tails}.
\end{cases}
\end{equation*}

\item Update $L_{i^*}\leftarrow L_{i^*}\left(\frac{p+\eps}{p-\eps}\right)^{b_{i^*}}\left(\frac{1-p-\eps}{1-p+\eps}\right)^{1-b_{i^*}}$.
\end{enumerate}
\item Output the coin $i$ with maximum $L_i$.
\end{enumerate}

\end{tt}
\end{minipage}
}}
\end{center}


\section{Preliminaries}
Our proof of optimality is based on an optimal strategy for multitoken Markov games. We now formally define the multitoken Markov game and state the optimal strategy that has been studied for this game. We use the notation and results from \cite{online-optimal-strategy-DTW03}.

A {\em Markov system} $S=(V,P,C,s,t)$ consists of a state space $V$, a transition probability function $P:V\times V\rightarrow [0,1]$, a positive real cost $C_v$ associated with each state $v$, a start state $s$ and a target state $t$. Let $v(0),v(1),\ldots,v(k)$ denote a set of states taken by following the Markov system for $k$ steps. The cost of such a trip on $S$ is the sum $\sum_{i=0}^{k-1}C_{v(i)}$ of the costs of the exited states. 

Let $S_1,\ldots,S_n$ be $n$ Markov systems, each of which has a token on its starting state. A {\em simple multitoken Markov game} $G=S_1\circ S_2\circ \cdots\circ S_n$ consists of a succession of steps in which we choose one of the $n$ tokens, which takes a random step in its system (i.e., according to its $P_i$). After choosing a token $i$ on state $u$ say, we pay the cost $C_i(u)$ associated with the state $u$ of the system $S_i$. We terminate as soon as one of the tokens reaches its target state for the first time. A strategy denotes the policy employed to pick a token given the state of the $n$ Markov systems. The cost of such a game $\E[G]$ is the minimum expected cost taken over all possible strategies. The strategy that achieves the minimum expected cost is said to be {\em optimal}. A strategy is said to be {\em pure} if the choice of the token at any step is deterministic (entirely determined by the state of all Markov systems).

\begin{theorem}\label{thm:pure-optimal-strategy}\cite{online-optimal-strategy-DTW03}
Every Markov game has a pure optimal strategy. 
\end{theorem}
For any strategy $\pi$ for a Markov game $G$, we denote the expected cost incurred by playing $\pi$ on $G$ by $\E_{\pi}[G]$.

The pure optimal strategy in the multitoken Markov game is completely determined by the {\em grade} $\grade$ of the states of the systems. The grade $\gamma$ of a state is defined as follows: Given a Markov system $S=(V,P,C,s,t)$ and state $u$, let $S(u)=(V,P,C,u,t)$ denote the Markov system whose starting state is $u$. Consider the Markov game $S_g(u)$ -- where at any step of the game one is allowed to either play in $S(u)$ or quit. Quitting incurs a cost of $g$. Playing in $S(u)$ is equivalent to taking a step following the Markov system $S$ incurring the cost associated with the state of the system. The game stops once the target state is reached or once we quit. The grade $\grade(u)$ of state $u$ is defined to be the smallest real value $g$ such that there exists an optimal strategy $\sigma$ that plays in $S(u)$ in the first step. We note that, by definition, the cost of the game $S_{\grade(u)}(u)$ is $\E[S_{\grade(u)}(u)]=\grade(u)=\E_{\sigma}[S_{\grade(u)}(u)]$.

\begin{theorem}\label{thm:gittins-strategy}\cite{online-optimal-strategy-DTW03}
Given the states $u_1,\ldots,u_n$ of the Markov systems in the multitoken Markov game, the unique optimal strategy is to pick the token $i$ such that $\grade(u_i)$ is minimal.
\end{theorem}
We observe that the above results can be extended in a straightforward manner to the case where (1) the number of Markov systems is countably infinite, i.e., $n=\infty$ and (2) the Markov systems have infinite state space but all states are locally finite (i.e., the number of possible transitions from any fixed state is finite), by working through the proofs in \cite{online-optimal-strategy-DTW03}. The Markov systems that will be considered for our purpose will satisfy these two properties.

We use the following results from \cite{gamblers-ruin-probability-KE02} to bound the number of tosses. 
\begin{theorem}\label{thm:gamblers-ruin}\cite{gamblers-ruin-probability-KE02}
Let $X\in [-\nu,\mu]$ be the random variable that determines the step-sizes of a one dimensional random walk with absorbing barriers at $-L$ and $W$ such that $\prob{X>0}>0$, $\prob{X<0}>0$, $\expec{X}\neq 0$. Let $L^*=L+\nu$, $W^*=W+\mu$ and $\phi(\rho):=\expec{\rho^{X}}$. 
\begin{enumerate}
\item The function $\phi(\rho)$ is convex. If $\expec{X}\ne 0$, there exists a unique $\rho_0\in (0,1)\cup(1,\infty)$ such that $\phi(\rho_0)=1$. If $\expec{X}<0$, then $\rho_0>1$ and if $\expec{X}>0$, then $\rho_0<1$.
\item 
\[
\prob{\text{Absorption at $W$}}\ge \frac{1-\rho_0^{L}}{1-\rho_0^{L+W^*}}.
\]
\item If $\expec{X}<0$, then 
\[
\expec{\text{Number of steps to absorption}}\le \frac{L^*}{\abs{\expec{X}}}.
\]
\item If $\expec{X}>0$, then 
\[
\expec{\text{Number of steps to absorption}}\le \frac{\left(L+W^*\right)}{\expec{X}}\left(\frac{1-\rho_0^{L^*}}{1-\rho_0^{L^*+W}}\right).
\]
\end{enumerate}
\end{theorem}

\section{Correctness}
We first argue the correctness of the algorithm.

\begin{lemma}\label{lem:heavy-coin-large-likelihood}
Given the history $D_i$ for a coin $i$, 
\[
\prob{\text{Coin $i$ is heavy}|D_i}\ge 1-\delta \text{ if and only if } L_i\ge \left(\frac{1-\delta}{\delta}\right)\left(\frac{1-\alpha}{\alpha}\right).
\]
\end{lemma}
\begin{proof} 
The lemma is a straightforward application of Bayes' theorem.
\begin{align*}
\prob{\text{Coin $i$ is heavy}|D_i}&= \frac{\prob{D_i|\text{Coin $i$ is heavy}}\prob{\text{Coin $i$ is heavy}}}{\prob{D_i}}\\
&=\frac{\alpha(p+\eps)^{h_i}(q-\eps)^{t_i}}{\alpha(p+\eps)^{h_i}(q-\eps)^{t_i}+(1-\alpha)(p-\eps)^{h_i}(q+\eps)^{t_i}}\\
&=\frac{\alpha L_i}{\alpha L_i+(1-\alpha)}.
\end{align*}
Thus, it follows that
\[
\prob{\text{Coin $i$ is heavy}|D_i}\ge 1-\delta \text{ if and only if } L_i\ge \left(\frac{1-\delta}{\delta}\right)\left(\frac{1-\alpha}{\alpha}\right).
\]
\end{proof}
The algorithm computes the likelihood ratio $L_i$ for each coin $i$ based on the history of outcomes of the coin. The algorithm repeatedly tosses coins until there exists $i^*$ such that $L_{i^*}\ge(1-\alpha)(1-\delta)/\alpha\delta$. Thus, if $i^*$ is output by Algorithm Likelihood-Toss, then 
\[
\prob{\text{Coin $i^*$ is heavy}|D_{i^*}}\ge 1-\delta.
\]

\section{Optimality of the Algorithm}
Consider the log-likelihood of a coin $i$ defined as $X_i:=\log{L_i}$. Given the history of a coin, the log-likelihood of the coin is determined uniquely. In the beginning, the history is empty and hence all log-likelihoods are identically zero. The influence of a toss on the log-likelihood is a random step for $X_i$ -- if the outcome of the toss is a head, then $X_i\leftarrow X_i+\Delta_H$ and if the outcome is a tail, then $X_i\leftarrow X_i-\Delta_T$. Thus, the toss outcomes of the coin leads to a 1-dimensional random-walk of the log-likelihood function associated with the coin. Further, since we stop tossing as soon as the log-likelihood of a coin is greater than $B=\log{(1-\alpha)(1-\delta)/\alpha\delta}$, the random-walk has an absorbing barrier at $B$. We observe that the random walks performed by the coins are independent of each other since each coin being heavy is independent of the rest of the coins. 

Thus, we have infinitely many identical Markov systems $S_1,S_2,\ldots, $ with each one starting in state $X_i=0$.
Each Markov system also has a target state, namely the boundary $B$. A strategy to pick a coin to toss is equivalent to picking a Markov system $i$. Each toss outcome is equivalent to the corresponding system taking a step following the transition probability and step size of the system. The goal to minimize the expected number of future tosses is equivalent to minimizing the expected number of steps for one of the Markov systems to reach the target state.

Therefore, we are essentially seeking an optimal strategy to play a multitoken Markov game. We show that the strategy employed by Algorithm Likelihood-Toss is an optimal strategy to play the multitoken Markov game that arises in our setting.

Let the Markov system associated with the one-dimensional random walk of the log-likelihood function of the history of the coin be $S=(V,P,C,s,t)$. Here, the state space $V$ consists of every possible real value that is at most $B$. The target state is a special state determined by $t=B$. The starting state is $s=0$. Given the current state $X$, the transition cost incurred is one while transition probabilities are defined as follows:
\begin{equation*}
X\rightarrow
\begin{cases}
& \min\{X+\Delta_H,B\}\text{ with probability $\prob{\text{Heads}|X}$},\\
& X-\Delta_T\text{ with probability $1-\prob{\text{Heads}|X}$}
\end{cases}
\end{equation*}
where 
\begin{align*}
\prob{\text{Heads}|X}&=\prob{\text{Heads}|\text{Heavy coin}}\prob{\text{Heavy coin}|X}\\
&\quad \quad +\prob{\text{Heads}|\text{Non-heavy coin}}\prob{\text{Non-heavy coin}|X}\\
&=\frac{(p+\eps)\alpha e^{X}}{\alpha e^{X}+(1-\alpha)}+\frac{(p-\eps)(1-\alpha)}{\alpha e^{X}+(1-\alpha)}.
\end{align*}
We observe that the transition probabilities in this random-walk vary with the state of the system (as opposed to the well-known random-walk under uniform transition probability). It is clear that this Markov system is locally finite -- the number of possible states reachable using one transition from any fixed state is only two. In this modeling of the Markov System for the log-likelihood of each coin, we do not condition on the coin being heavy or not-heavy. We are postponing this decision by conditioning based on the history.

\subsection{Proof of Optimality}
We now show that the grade is a monotonically non-increasing function of the log-likelihood.
\begin{lemma}\label{lem:grade-likelihood-monotonicity}
Consider the Markov System $S=(V,P,C,s,t)$ associated with the log-likelihood function. Let $X,Y\in V$ such that $X\ge Y$. Then $\grade(X)\le \grade(Y)$.
\end{lemma}
\begin{proof} 
Let $\grade(Y)=g$. Then, by definition of grade, it follows that there exists a pure optimal strategy $\sigma$ that chooses to toss the coin in the first step in $S_{g}(Y)$ and $E_{\sigma}[S_{g}(Y)]=g$. We will specify a mixed strategy $\pi$ for $S_{g}(X)$ such that $\E_{\pi}[S_{g}(X)]\le g$ and $\pi$ chooses to play in the system $S(X)$ in the first step. It follows by definition that $\grade(X)\le g$.

The pure strategy $\sigma$ can be expressed by a (possibly infinite) binary decision tree $D_{\sigma}$ as follows: Each node $u$ has an associated label $l(u)\in \reals$. Each edge has a label from $\{H,T\}$. The root node $v$ is labeled $l(v)=Y$. On reaching $l(u)<B$, if $\sigma$ chooses to play in the system, then $u$ has two children - the left and right children $u_L, u_R$ are labeled $l(u_L)=l(u)+\Delta_H$ and $l(u_R)=l(u)-\Delta_T$ respectively. The edges $(u,u_L), (u,u_R)$ are labeled $H$ and $T$ respectively. On reaching $l(u)<B$, if $\sigma$ decides to quit, then $u$ is a leaf node. Finally, if $l(u)\ge B$, then $u$ is a leaf node. We observe that since $\sigma$ plays in the system $S_g(Y)$ in the first step, the root of $D_{\sigma}$ is not a leaf. (See Figure \ref{figure:trees} for an example.)

\begin{figure}[ht]
\centering
\begin{tabular}{cc}
\includegraphics[scale=0.5]{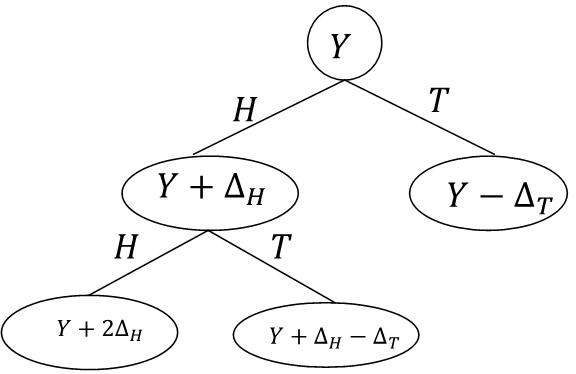} &
\includegraphics[scale=0.5]{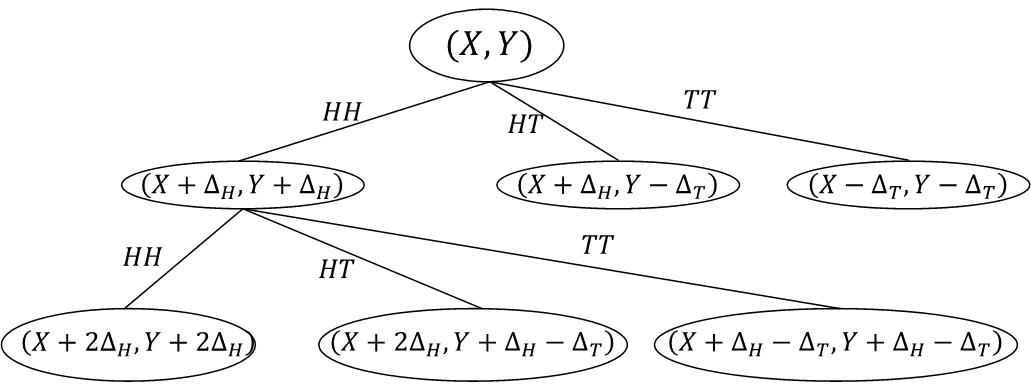}\\
Binary Decision Tree $D_{\sigma}$ & Ternary Decision Tree $D_{\pi}$
\end{tabular}
\caption{An example of a strategy $\sigma$ represented as a binary decision tree $D_{\sigma}$ for the Markov game $S_g(Y)$ where $B-2\Delta_H<Y<B-\Delta_H$; the strategy $\sigma$ is to continue playing in the system $S_g(Y)$ on reaching states $Y$ and $Y+\Delta_H$ and to quit on reaching states $Y-\Delta_T$ and $Y+\Delta_H-\Delta_T$. The corresponding ternary decision tree $D_{\pi}$ derived from $D_{\sigma}$ is also shown.}
\label{figure:trees}
\end{figure}

We obtain a mixed strategy $\pi$ for $S_g(X)$ by considering the following ternary tree $D_{\pi}$ derived from $D_{\sigma}$: Each node $u$ in $D_{\pi}$ has an associated label $(l_X(u),l_Y(u))\in \reals^2$. Each edge in $D_{\pi}$ has a label from $\{HH,HT,TT\}$. There is an onto mapping $m(u)$ from each node $u\in D_{\pi}$ to a node in $D_{\sigma}$. The root node $u$ is labeled $(l_X(u)=X,l_Y(u)=Y)$ and $m(u)=$Root$(D_{\sigma})$. For any node $u$, if $m(u)$ is a leaf, then $u$ is a leaf. Let $u$ be a node such that $v=m(u)$ is not a leaf. Let $v_H$ and $v_T$ denote the left and right children of $v$. Create children $u_{HH}, u_{HT}, u_{TT}$ as nodes adjacent to edges labeled $HH, HT, TT$ respectively. Define the mapping $m(u_{HH})=v_H$, $m(u_{HT})=v_T$, $m(u_{TT})=v_T$ and set 
\begin{align*}
l_X(u_{HH})=l_X(u)+\Delta_H,\ l_X(u_{HT})&=l_X(u)+\Delta_H,\ l_X(u_{TT})=l_X(u)-\Delta_T,\\
l_Y(u_{HH})=l_Y(u)+\Delta_H,\ l_Y(u_{HT})&=l_Y(u)-\Delta_T,\ l_Y(u_{TT})=l_Y(u)-\Delta_T.
\end{align*}
By construction of $D_{\pi}$, it follows that if $X\ge Y$, then at any node $u$ in $D_{\pi}$, $l_X(u)\ge l_Y(u)$ and hence, $\prob{\text{Heads}|l_X(u)}\ge \prob{\text{Heads}|l_Y(u)}$.

Our mixed strategy $\pi$ for $S_g(X)$ is based on $D_{\pi}$. The strategy at any step maintains a pointer to some node $u$ in $D_{\pi}$. Initialize the pointer to the root node $u$. If the pointer is at a non-leaf node $u$, then $\pi$ chooses to play in the system. If the step in the system is a backward step (outcome of coin toss is a tail), then $\pi$ moves the pointer to $u_{TT}$. If the step in the system is a forward step (outcome of coin toss is a head), then $\pi$ generates a random number $r\in [0,1]$ and moves the pointer to the node $u_{HH}$ if $r<\prob{\text{Heads}|l_Y(u)}/\prob{\text{Heads}|l_X(u)}$ and to the node $u_{HT}$ if $r\ge\prob{\text{Heads}|l_Y(u)}/\prob{\text{Heads}|l_X(u)}$. If the pointer is at a leaf node $u$ such that $l_Y(u)<B$, then $\pi$ quits the system. Otherwise, $l_Y(u)\ge B$ and hence $l_X(u)\ge B$. Thus, the strategy $\pi$ is a valid mixed strategy for $S_g(X)$ and $\pi$ plays in the system $S_g(X)$ in the first step since $\sigma$ plays in the system $S_g(Y)$ in the first step.

It only remains to show that $\E_{\pi}[S_g(X)]\le g$. This is shown in Claim \ref{claim:expected-cost-of-pi}.
\end{proof}

\begin{claim}\label{claim:expected-cost-of-pi}
\[
\E_{\pi}[S_g(X)]\le g.
\]
\end{claim}
\begin{proof}
The cost of using $\sigma$ for $S_{g}(Y)$ can be simulated by running a random process in $D_{\sigma}$ and considering an associated cost. For each non-leaf node in $D_{\sigma}$ associate a cost of $1$ and for each leaf node $u$ in $D_{\sigma}$ such that $l(u)<B$, associate a cost of $g$. Consider the following random process $RP_1(u)$ for a node $u\in D_{\sigma}$: Begin at node $u$ of $D_{\sigma}$. On reaching a non-leaf node $v$, repeatedly traverse the tree $D_{\sigma}$ by taking the left child with probability $\prob{\text{Heads}|l(v)}$ and the right child with the remaining probability until a leaf node is reached. The cost of the random process is the sum of the cost incurred along the nodes in the path traversed by the random process. Let $\E[D_{\sigma}(u)]$ denote the expected cost. Then, by construction of $D_{\sigma}$, it follows that $\E[D_{\sigma}(r)]=\E_{\sigma}[S_{g}(l(r))]=g$ for the root node $r$ in $D_{\sigma}$.

Next, we give a random process $RP_2$ on $D_{\pi}$ that relates the expected cost of following strategy $\pi$ on $S_g(X)$ and the expected cost of following strategy $\sigma$ on $S_g(Y)$. We first associate a cost with each node $u$ in $D_{\pi}$: For each non-leaf node $u$, if $l_X(u)<B$, then cost $c_X(u)=1$, and if $l_Y(u)<B$, then cost $c_Y(u)=1$. For each leaf node $u$, if $l_X(u)<B$, then cost $c_X(u)=g$ and if $l_Y(u)<B$, then cost $c_Y(u)=g$. The remaining costs are zero. Here, we observe that $c_X(u)\le c_Y(u)$ for every node $u\in D_{\pi}$.

We define the random process $RP_2(v)$ for a node $v\in D_{\pi}$ as follows: Begin at node $v$ and repeatedly traverse the tree $D_{\pi}$ by taking one of the three children at each non-leaf node until a leaf node is reached. On reaching a non-leaf node $u$, traverse to $u_{HH}$ with probability $\prob{\text{Heads}|l_Y(u)}$, to $v_{HT}$ with probability  $\prob{\text{Heads}|l_X(u)}-\prob{\text{Heads}|l_Y(u)}$ and to $u_{TT}$ with the remaining probability. Let $P(v)$ be the set of nodes in the path traversed by the random process $RP_2(v)$. Let the cost incurred be $\overline{c}_X(v)=\sum_{u\in P(v)}c_X(u)$ and $\overline{c}_Y(v)=\sum_{u\in P(v)}c_Y(u)$. Now, the cost incurred by following strategy $\pi$ for $S_g(X)$ is the same as the cost $\overline{c}_X(r)$ incurred by the random process $RP_2(r)$, where $r$ is the root node in $D_{\pi}$. 

By construction of $D_{\pi}$ from $D_{\sigma}$, it follows that for each node $v\in D_{\pi}$, the expected cost $\overline{c}_Y(v)$ of the random process $RP_2(v)$ is equal to the expected cost of the random process $RP_1(m(v))$.
Hence, $\E[\overline{c}_Y(r)]=\E[D_{\sigma}(m(r))]=g$ for the root node $r$ in $D_{\pi}$. 
Next, since $c_X(u)\le c_Y(u)$ for every node $u$, it follows that $\E[\overline{c}_X(r)]\le \E[\overline{c}_Y(r)]=g$.
Finally, the expected cost incurred by following mixed strategy $\pi$ for $S_g(X)$ is exactly equal to $\E[\overline{c}_X(r)]$.
\end{proof}

\begin{proof}[Proof of Theorem \ref{thm:algorithm-optimality}]
We use Algorithm Likelihood-Toss. By Lemma \ref{lem:heavy-coin-large-likelihood}, the optimal adaptive strategy also  minimizes the expected number of tosses to identify a coin $i$ such that the log-likelihood $X_i\ge B$.

The strategy adopted by Algorithm Likelihood-Toss at any stage is to toss the coin with maximum log-likelihood. Let the Markov system associated with the one-dimensional random walk of the log-likelihood function of the history of the coin be $S=(V,P,C,s,t)$. We have infinitely many independent and identical Markov systems $S_1=S_2=\ldots=S$ associated with the log-likelihood function of the respective coin. By Theorem \ref{thm:gittins-strategy}, the optimal strategy to minimize the expected number of tosses to identify a coin $i$ such that the log-likelihood $X_i\ge B$ is to toss the coin $i$ such that $\grade(X_i)$ is minimal. Lemma \ref{lem:grade-likelihood-monotonicity} shows that the grade function $\grade(X)$ is monotonically non-increasing. Thus, tossing the coin with maximum log-likelihood is an optimal strategy.

By the description of the algorithm, it is clear that the algorithm starts tossing a fresh/new coin only if the log-likelihood of the current coin decreases below zero. The time to update the likelihood ratio of the current coin after a coin toss is only a constant and hence the time to identify the coin to toss is $O(1)$.

\end{proof}

\section{Number of Coin Tosses}\label{sec:quantification}
In this section, we give an upper bound on the number of coin tosses performed by Algorithm Likelihood-Toss. The algorithm repeatedly tosses a coin while the log-likelihood of the coin is at least zero and starts with a fresh coin if the log-likelihood of the coin is less than zero. The algorithm terminates if the log-likelihood of a coin is at least $B$.

Consider the random walk of the log-likelihood function. The random walk has absorbing barriers at $B$ and at every state less than $0$. 
\begin{lemma}\label{lem:absorption-bounds}
Let $C$ and $D$ denote the expected number of tosses to get absorbed for a non-heavy and heavy coin respectively. Let $\pi$ denote the probability that a heavy coin gets absorbed at $B$. Then, under the assumptions of Theorem \ref{thm:no-of-tosses},
\begin{enumerate}
\item 
\[
\pi\ge \frac{\Delta_H(p+\eps)-\Delta_T(q-\eps)}{2(\Delta_H+\Delta_T)}.
\]
\item 
\[
\frac{D}{\pi}\le \left(\frac{8B}{\Delta_H(p+\eps)-\Delta_T(q-\eps)}\right)\left(\frac{\Delta_H+\Delta_T}{\Delta_H(p+\eps)}\right).
\]
\item 
\[
C\le \frac{2(\Delta_H+\Delta_T)}{\Delta_T(q+\eps)-\Delta_H(p-\eps)}.
\]
\end{enumerate}
\end{lemma}
\begin{proof}
Consider a modified random walk where the starting state is $\Delta_H$ as opposed to zero. Let $C'$ and $D'$ denote the expected number of tosses for the modified walk to get absorbed using a non-heavy and heavy coin respectively. Let $\pi'$ denote the probability that the modified walk gets absorbed at $B$ using a heavy coin. Then, $D\le D'+1\le 2D'$, $C\le C'+1\le 2C'$, $\pi=(p+\eps)\pi'$. 

We use Theorem \ref{thm:gamblers-ruin}. For the modified random walk, we have that $L=\Delta_H$, $W=B-\Delta_H$, $\nu=\Delta_T$, $\mu=\Delta_H$. For the modified random walk using a heavy coin, the step sizes are
\begin{equation*}
X=
\begin{cases}
& \Delta_H\text{ with probability $p+\eps$}\\
& -\Delta_T\text{ with probability $q-\eps$},
\end{cases}
\end{equation*}
and for the modified random walk using a non-heavy coin, the step sizes are 
\begin{equation*}
Y=
\begin{cases}
& \Delta_H\text{ with probability $p-\eps$}\\
& -\Delta_T\text{ with probability $q+\eps$},
\end{cases}
\end{equation*}
For $\eps>0$, we have that $\expec{Y}<0$. Therefore, 
\[
C'\le \frac{\Delta_H+\Delta_T}{\Delta_T(q+\eps)-\Delta_H(p-\eps)}
\]
and hence we have the bound on $C$.

Now consider the modified random walk using a heavy coin. For $\eps>0$, we have that $\expec{X}>0$. Let $\rho_0<1$ be the unique real value such that $\expec{\rho_0^X}=1$. Thus,
\begin{align*}
\pi'&\ge \frac{1-\rho_0^{\Delta_H}}{1-\rho_0^{B+\Delta_H}}\\
D'&\le \frac{(\Delta_H+B)}{\expec{X}}\left(\frac{1-\rho_0^{\Delta_H+\Delta_T}}{1-\rho_0^{B+\Delta_T}}\right).
\end{align*}
Since $\phi(\rho)$ is convex, it can be shown that the minimum value of $\phi(\rho)$ occurs at 
\[
\rho_{\min}=\left(\frac{\Delta_T(q-\eps)}{\Delta_H(p+\eps)}\right)^{\frac{1}{\Delta_H+\Delta_T}}
\]
and hence, $\rho_0< \rho_{\min}<1$. Thus,
\begin{align*}
\frac{D'}{\pi'}&\le \frac{(\Delta_H+B)}{\expec{X}}\left(\frac{1-\rho_0^{B+\Delta_H}}{1-\rho_0^{B+\Delta_T}}\right)\left(\frac{1-\rho_0^{\Delta_H+\Delta_T}}{1-\rho_0^{\Delta_H}}\right)\\
&\le \frac{2B}{\expec{X}}\left(\frac{1-\rho_0^{\Delta_H+\Delta_T}}{1-\rho_0^{\Delta_H}}\right)\quad \quad \text{(by the assumption $\delta<\delta_0$)}\\
&< \frac{2B}{\expec{X}}\left(\frac{1-\rho_{\min}^{\Delta_H+\Delta_T}}{1-\rho_{\min}^{\Delta_H}}\right) \quad \quad \text{(since $\rho_0<\rho_{\min}$)}\\
&=\frac{2B}{\Delta_H(p+\eps)}\left(\frac{1}{1-\left(\frac{\Delta_T(q-\eps)}{\Delta_H(p+\eps)}\right)^{\frac{\Delta_H}{\Delta_H+\Delta_T}}}\right)\\
&\le \frac{4B(\Delta_H+\Delta_T)}{\expec{X}\Delta_H}.
\end{align*}
and we obtain the bound on the ratio $D/\pi$. Finally, to lower bound $\pi'$, we observe that
\begin{align*}
\pi'&\ge \frac{1-\rho_0^{\Delta_H}}{1-\rho_0^{B+\Delta_H}}\\
&\ge \frac{1-\rho_{\min}^{\Delta_H}}{1-\rho_{\min}^{B+\Delta_H}}\\
&\ge 1-\rho_{\min}^{\Delta_H}\\
&\ge \frac{\expec{X}}{2(\Delta_H+\Delta_T)(p+\eps)}.
\end{align*}
\end{proof}

\begin{proof}[Proof of Theorem \ref{thm:no-of-tosses}].
We use Algorithm Likelihood-Toss. Consider the one-dimensional random walk of the log-likelihood function. The random walk has absorbing barriers at $B$ and at every state less than $0$. Let $C$ and $D$ denote the expected number of tosses to get absorbed for a non-heavy and heavy coin respectively. Let $\pi$ denote the probability that a heavy coin gets absorbed at $B$. Let $D_0$ and $D_1$ denote the expected number of tosses of a heavy coin to get absorbed at $0$ and $B$ respectively. Then, $D=(1-\pi)D_0+\pi D_1$.

Let $E$ denote the expected number of tosses performed by algorithm Likelihood-Toss. Then,
\begin{align*}
E&\le (1-\alpha)(C+E)+\alpha((1-\pi)(D_0+E)+\pi D_1)\\
\Rightarrow E&\le \frac{(1-\alpha)}{\alpha}\frac{C}{\pi} + \frac{D}{\pi}.
\end{align*}
By Lemma \ref{lem:absorption-bounds}, we have that
\begin{align*}
E&\le
\left(\frac{4(\Delta_H+\Delta_T)}{\Delta_H(p+\eps)-\Delta_T(q-\eps)}\right)
\left(\left(\frac{1-\alpha}{\alpha}\right)\left(\frac{\Delta_H+\Delta_T}{\Delta_T(q+\eps)-\Delta_H(p-\eps)}\right)+
\left(\frac{2B}{\Delta_H(p+\eps)}\right)
\right). 
\end{align*}
The final upper bound follows by substituting for $\Delta_H, \Delta_T$ and $B$ and using the following inequalities (derived by straightforward calculus),
\begin{align*}
\frac{2}{\eps}&\ge\max\left\{\frac{\Delta_H+\Delta_T}{\Delta_H(p+\eps)-\Delta_T(q-\eps)}, \frac{\Delta_H+\Delta_T}{\Delta_T(q+\eps)-\Delta_H(p-\eps)}\right\},\\
\Delta_H&\ge \frac{\eps}{p-\eps}.\\
\end{align*}
\end{proof}

\section{Discussion}
We gave an adaptive strategy that tosses coins in order to achieve a certain stopping condition, namely, the existence of a coin whose posterior probability of being heavy is at least a given threshold. Our strategy has minimum cost where cost is measured by the expected number of future tosses by following the strategy to attain the stopping condition. We achieved this by performing the best possible action after observing the outcome of each coin toss. We note that our algorithm can also be modified to start from any fixed history of outcomes by appropriately modifying the initialization step. The optimality of the action is exhibited using tools from the field of Markov games. A major limitation of our algorithm is that it is optimal only in the setting where the coins are independently heavy and non-heavy. It would be very interesting to devise an adaptive strategy where the coins are not necessarily independent -- say we have $n$ coins with exactly one heavy coin and the goal is to attain the stopping condition. In this setting, we note that the posterior probability of a fixed coin being heavy depends on the outcomes of the tosses of all the coins and not just any fixed coin. \\

\noindent {\bf Acknowledgment}. We thank Santosh Vempala for valuable comments.

\bibliographystyle{abbrv}
\bibliography{references}
\end{document}